\documentclass[runningheads]{llncs}

\usepackage{microtype}

\usepackage[normalem]{ulem}
\usepackage[utf8]{inputenc}

\usepackage{mathtools} 
\usepackage{wrapfig} 

\usepackage{xcolor}
\usepackage{xspace}
\usepackage{fmtcount} 
\usepackage{xstring}
\usepackage{cite}
\usepackage{amssymb}

\usepackage{nicefrac}
\newcommand{\nf}[2]{\ensuremath{\nicefrac{#1}{#2}}}

\usepackage{hyperref}

\newcommand{\ie}{i.\,e.\xspace}
\newcommand{\eg}{e.\,g.,\xspace}
\newcommand{\Wlog}{W.\,l.\,o.\,g.\xspace}

\newcommand{\fire}{fire\xspace}
\newcommand{\burn}{consumption\xspace} 
\newcommand{\ratio}{consumption-ratio\xspace}
\newcommand{\crews}{barrier system\xspace}

\newcommand{\wall}{barrier\xspace}
\newcommand{\walls}{barriers\xspace}

\newcommand{\BF}[1]{\mathcal{C}_{#1}}
\newcommand{\QF}[1]{\mathcal{Q}_{#1}}

\newcommand{\consumptionRatio}[1][]{\ensuremath{\QF{}^{#1}(t)}\xspace}
\newcommand{\A}{a}
\newcommand{\B}{b}
\newcommand{\C}{c}
\newcommand{\D}{d}
\newcommand{\startWall}{\ensuremath{s}\xspace}

\begin{document}
\title{Geometric Firefighting in the Half-plane\thanks{This work has been supported by DFG grant Kl 655/19 as part of a DACH project.} \thanks{This is a pre-print of an article published in \textit{Algorithms and Data Structures - 16th International Symposium, {WADS} 2019}.
}}
%
%
\author{Sang-Sub Kim
	\and Rolf Klein
	\and David Kübel
	\and Elmar Langetepe
	\and Barbara Schwarzwald
}
%
\authorrunning{S. Kim, R. Klein, D. Kübel, E. Langetepe, and B. Schwarzwald}
%
\institute{Department of Computer Science, University of Bonn, 53115 Bonn, Germany \\
	\email{\{sang-sub,rolf.klein,dkuebel,schwarzwald\}@uni-bonn.de}\\
	\email{elmar.langetepe@cs.uni-bonn.de}}
\maketitle              

\begin{abstract}

In 2006, Alberto Bressan~\cite{b-dicff-07} suggested the following problem. Suppose a circular fire spreads in the Euclidean plane at unit speed. The task is to build, in real time, barrier curves to contain the fire. At each time $t$ the total length of all barriers built so far must not exceed $t \cdot v$, where $v$ is a speed constant. How large 
a speed $v$ is needed? He proved that speed $v>2$ is sufficient, and that $v>1$ is necessary.
This gap of $(1,2]$ is still open. The crucial question seems to be the following. {\em When trying to contain a fire, should one build, at maximum speed, the enclosing barrier, or does it make sense to spend some time on placing extra delaying barriers in the fire's way?}
We study the situation where the fire must be contained in the upper $L_1$ half-plane by an infinite horizontal barrier to which vertical line segments may be attached as delaying barriers.
Surprisingly, such delaying barriers are helpful when properly placed. We prove that speed $v=1.8772$ is sufficient, while $v >1.66$ is necessary.

\keywords{barrier, firefighting, geodesic circle}

\end{abstract}
\section{Introduction and problem statement}
Fighting wildfires is a difficult problem, involving many parameters one can neither foresee nor control. But there seem to be two main techniques firefighters employ, namely to extinguish the fire by dropping water or chemicals from aircraft, and to prevent the fire from spreading further by firebreaks.
In 2006, Alberto Bressan~\cite{b-dicff-07} developed a rather general model for containing a fire by means of barrier curves that must be built in real time, subject to velocity constraints. Barriers are impenetrable by fire, they do not burn and cannot be moved once built.

In addition to general optimality results \cite{bbfj-bsfcp-08,bw-msbph-09,bw-osibp-12}, in~\cite{b-dicff-07} Bressan proposed the following problem. Suppose a circular fire spreads in the plane at unit speed. In real time, barrier curves must be built to contain it. At each time $t$, the total length of barriers built so far must not exceed $t$ times $v$, for some velocity constant $v$.
The question is how large a velocity is needed to contain the fire.
Bressan showed that $v>1$ is necessary and that $v>2$ is sufficient; see also~\cite{kl-cgc-16} for short proofs. 
He conjectured that speed $v=2$ is necessary.
But the gap $(1,2]$ is still open, even though a 500 USD reward has been offered~\cite{b-awa-11} in 2011.

It seems that the difficulty lies with the following question. {\em To contain a fire, should one build an enclosing barrier at maximum speed, or is it better to invest some time in building extra delaying barriers that will not be part of the final enclosure but can slow the fire down during construction?} 
If delaying barriers could be shown to be useless, Bressan's proof of the lower bound~1 could be easily extended to prove his conjecture, the lower bound of~2.
In fact they consider a special variant in \cite{bw-msbph-09}, where the fire spreads in a half plane. In that case they can construct an optimal strategy without delaying barriers, that encloses the fire between the boundary of the half plane and the barrier curve.

To study the effectiveness of delaying barriers we study a different setting where an infinite horizontal \wall has to be built to contain the fire in the upper half-plane, instead of the interior of a closed barrier curve. 
To this horizontal \wall, vertical line segments may be attached as delaying barriers. Without vertical barriers speed $v=2$ is necessary and sufficient to build the horizontal \wall.
While it takes extra time to build vertical barriers, they offer some respite because the expanding fire has to overcome them before it reaches the horizontal \wall again. 
To simplify matters further we are working in the $L_1$ norm, so that distances are free of square roots. Also, all intersections of the fire's boundary with the barriers advance at unit speed.

Our main result is the following. In our setting, speed $v>1.66$ is necessary, and, 
with a careful placement of delaying barriers, speed $v=1.8772$ is sufficient. 
While this result does not disprove Bressan's conjecture it casts a new light on the problem by showing that building delaying barriers can be helpful. 
Also, the gap we leave open is smaller than the one for the original containment problem. 

Previous, but weaker results have been presented at EuroCG'18~\cite{kkls-ffp-18}.

\subsection{Related work}
Among theoretical work on {\em extinguishing} a fire, the ``lion and man'' problem stands out~\cite{dsz-ovlmp-08,bkns-eosit-09,bggk-hmlnc-09,k-rpffp-18}. Here, $r$ fighters are tasked with quenching a fire in an $n \times n$ grid. In every step, fighters and fire move simultaneously to adjacent cells, subject to certain rules.
While $r=n$ fighters can easily extinguish the fire, $\lfloor \nf{n}{2} \rfloor$ fighters are not enough.
The gap in between is still open, despite serious efforts.

How to {\em contain} a fire has received a lot of attention in graph theory, see, \eg \cite{fm-fpsrd-07,fkmr-fpgmd-07,fhl-fpgc-16}.
In quite a few examples, in each round, a stationary guard can be placed in a vertex not on fire, then the fire spreads to all unguarded adjacent vertices.
This continues until the fire cannot spread any further.
The problem to determine the maximum number of vertices that can be protected is NP-hard, even in trees of degree~3.

Similar in spirit is a geometric firefighting problem in simple polygons~\cite{kll-aagfb-18}, where barriers must be chosen from a set of pairwise disjoint diagonals, to save an area of maximum size. Even for convex polygons, the problem is NP-hard, but a 0.086 approximation algorithm exists.

It is interesting to see what happens when building a barrier along the boundary of an expanding circular fire~\cite{bbfj-bsfcp-08,bw-msbph-09,kll-ffp-15,kllls-ffp-18}. A spiraling curve results that closes on itself, and thus contains the fire, if the speed of building is larger than $2.6144$. Then the number of rounds to completion
can be determined by residue calculus. Below this threshold, the curve keeps winding forever.

The rest of this paper is organized as follows.
Section~\ref{section:modelDefinition} formally introduces the problem as well as terms and definitions required for the analysis.
In Section~\ref{section:lowerBounds} we develop a lower bound of $v > 1.66$.
In Section~\ref{section:upperBounds} we show that $v=\nf{17}{9}=1.\overline{8}$ is sufficient and discuss how this value can even be reduced to $v=1.8772$.

\section{Model}
\label{section:modelDefinition}
In our model, the \fire spreads from the origin and continuously expands over time with speed $1$ according to the $L_1$ metric.
To prevent the fire from immediately spreading into the lower half-plane, we allow an arbitrarily small head-start of barrier of length $\startWall$ into both directions along the $x$-axis.

\begin{figure}[ht]
	\centering
	\includegraphics[width=0.75\textwidth]{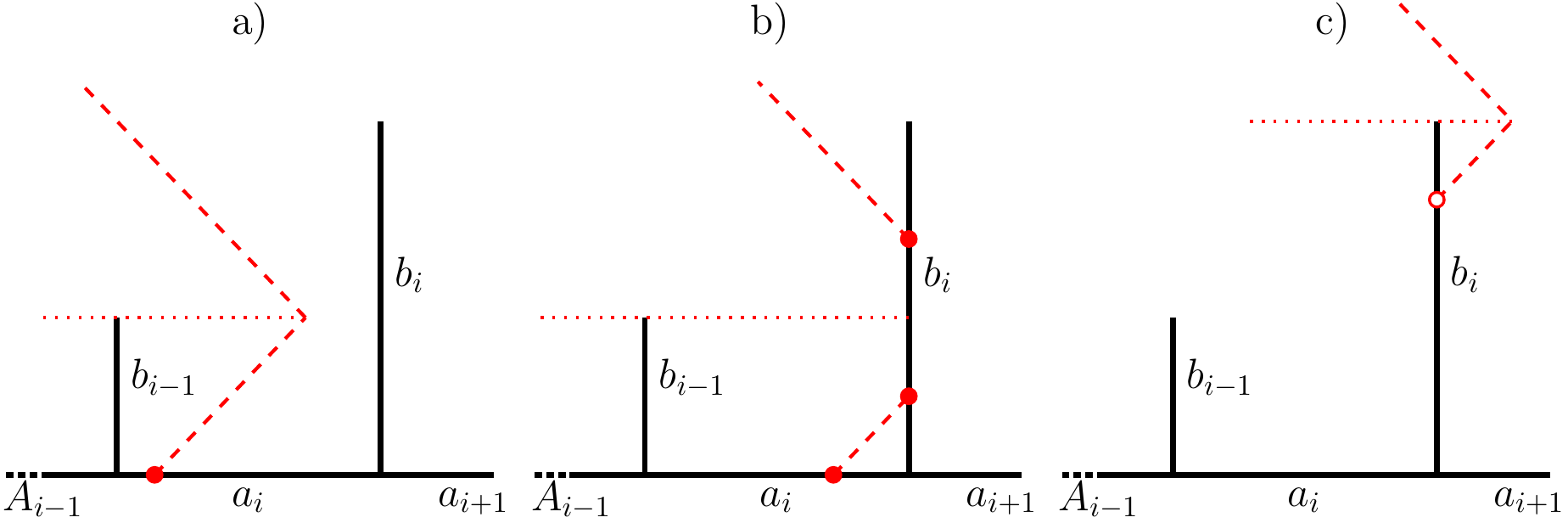}
	\caption{Fire spreading along delaying barriers. The dashed line shows the fire front at different times $t$, solid points represent \burn points, while empty points represent places, where the fire burns along the back of already consumed parts of the \wall $\B_i$. In a) there is one consumption point, so there is a $1$-interval in the right direction. In b) there are three \burn points and in c) there is a $0$-interval in the right direction as there are no consumption points.}
	\label{fig:delayingBarriers}
\end{figure}
Assume that a system of \walls has been built.
The \crews consists of a horizontal \wall containing the fire in the upper half-plane and several vertical delaying barriers attached to it.

To describe a \crews, we denote the $i$-th delaying \wall to the right by $b_i$.
The part of the horizontal \wall between $b_{i-1}$ and $b_i$ is denoted by $a_{i}$.
For simplicity, we also refer to their length by $a_i$ and $b_i$.
For the other direction, we use $c_i$ and $d_i$ respectively.
For convenience, $A_i := \sum_{j=1}^i \A_j$ will denote the total length of horizontal \walls in the right direction up till and including $\A_i$ and $B_i := \sum_{j=1}^i \B_j$ will denote the total length of vertical  \walls in the right direction up till and including $\B_i$.
Equivalently for the left direction we define $C_i$ and $D_i$.

As the fire spreads over the \crews, it represents a geodesic $L_1$ circle, which consumes the \walls when burning along them.
The fire-front is the set of all points in the plane, which shortest non-\wall-crossing path to the fire origin has length $t$.
We consider a point $x$ on a \wall as {\em consumed} at time $t$ if the \fire has reached this point at time $t$. That means there exists a non-\wall-crossing path of length at most $t$ from the fire origin to the point $x$.
Hence, any piece of the barrier is not consumed all at once, but as the fire burns along it.
We call a point on a \wall, which shortest non-\wall-crossing path to the fire has exactly length $t$ a \burn point at time $t$, so the \burn points are a subset of the fire front.
We call the number of \burn points at time $t$ the {\em current \burn} and a time interval with constant $k$ \burn points at all times a {\em $k$-interval}.

The fire front, \burn points and the effect of vertical delaying \walls are illustrated in \autoref{fig:delayingBarriers}.
As one can see, after the fire reaches a delaying barrier for the first time, it may burn along multiple \walls at multiple points.
However, after reaching both ends and passing the top of a \wall there might be no \burn for a while as the delaying \wall has already been burned along from the other side.

%
We define the \emph{total \burn} $\BF{}$ and \emph{\ratio} $\QF{}$ for a time interval $\left[t_1, t_2\right]$ in a \crews:
\begin{equation*}
\begin{array}{r l}
\BF{}(t_1, t_2) 
	& := \text{length of \wall pieces consumed by the fire between } t_1 \text{ and } t_2\\
\QF{}(t_1, t_2) 
	& := \frac{ \BF{}(t_1, t_2)}{t_2 - t_1}\,.
\end{array}
\end{equation*}

For the \burn in a time interval $[0, t]$, we will also write $\BF{}(t)$ and $\QF{}(t)$ for short.
In our setting, if $[t_0, t_1]$ is a $k$-interval, then  $\BF{}(t_1) = \BF{}(t_0) + (t_1-t_0)\cdot k$.

Note that all these definitions can easily be applied to either side of the \crews, denoted by \consumptionRatio[l], \consumptionRatio[r] and $\BF{}^l(t)$, $\BF{}^r(t)$ equivalently.
Obviously, $\consumptionRatio = \consumptionRatio[l] + \consumptionRatio[r]$ and $\BF{}(t) = \BF{}^l(t)+ \BF{}^r(t)$.

It is clear that when building a \crews simultaneously to the \fire spreading, then every piece of barrier should be build before the \fire reaches it. 
For a limited build speed $v$, it is necessary and sufficient to have $\BF{}(t) \leq v \cdot t$ for all times $t$, which means $v \geq \sup_t \QF{}(t)$.
The question then obviously is: What is the minimum speed $v$ for which such a \crews exists?

\section{Prerequisites}\label{sec:prereq}

Observe that a vertical \wall which is shorter than the predecessor in the same direction does not delay the fire. Hence, we can assume that vertical \walls in one direction increase strictly in length, so $\B_i > \B_{i-1}$ and $\D_i > \D_{i-1}$ for all $i>1$. But we can show an even stronger bound on the growth of successive vertical \walls.

\begin{lemma}\label{lemma:prereq}
If there exists a \crews with $\BF{}(t) \leq v \cdot t$ at all times $t$, then there also exists such a \crews in which any vertical \wall $\B_i$ (or $\D_i$) is more than twice as long as the previous \wall $\B_{i-1}$ (or $\D_{i-1}$) in the same direction.
\end{lemma}

\begin{proof}
	Assume we are given any \crews $\mathcal{S}$ with $\BF{S}(t) \leq v\cdot t$, not fulfilling both properties $\B_i > 2 \B_{i-1}$ and $\D_i > 2 \D_{i-1}$. Then we can transform it to a new \crews $\mathcal{S}'$ that fulfils both properties $\B_i > 2 \B_{i-1}$ and $\D_i > 2 \D_{i-1}$ while $\BF{S'}(t) < \BF{S}(t) \leq v\cdot t$ for all $t$.
	
	The construction is identical for both directions, so we just consider the right direction.
	Let $\B_k$ ($k>1$) be the first vertical \wall in the right direction with $\B_k < 2 \B_{k-1}$. Then we can remove $\B_k$ and move all following vertical \walls away from the fire by $2\Delta = 2(\B_{k} - \B_{k-1})$.
	So, more precisely the right side of our \crews $\mathcal{S}'$ consists of $\B'_i$ and $\A'_i$ as follows:
	\begin{eqnarray}
	\text{for } i<k \quad & \quad \B'_i = \B_i \quad & \quad \A'_i = \A_i \\
	\text{for } i = k \quad & \quad \B'_i = \B_{i+1} \quad & \quad \A'_i = \A_i + \A_{i+1} + 2\Delta \\
	\text{for } i \geq k \quad & \quad \B'_i = \B_{i+1} \quad & \quad \A'_i = \A_{i+1}
	\end{eqnarray}
	
	If $\B_k$ is the last vertical delaying \wall in the right direction, it can just be removed instead.
	\begin{figure}[htbp]%
		\begin{center}%
			\includegraphics[scale=0.7,keepaspectratio]{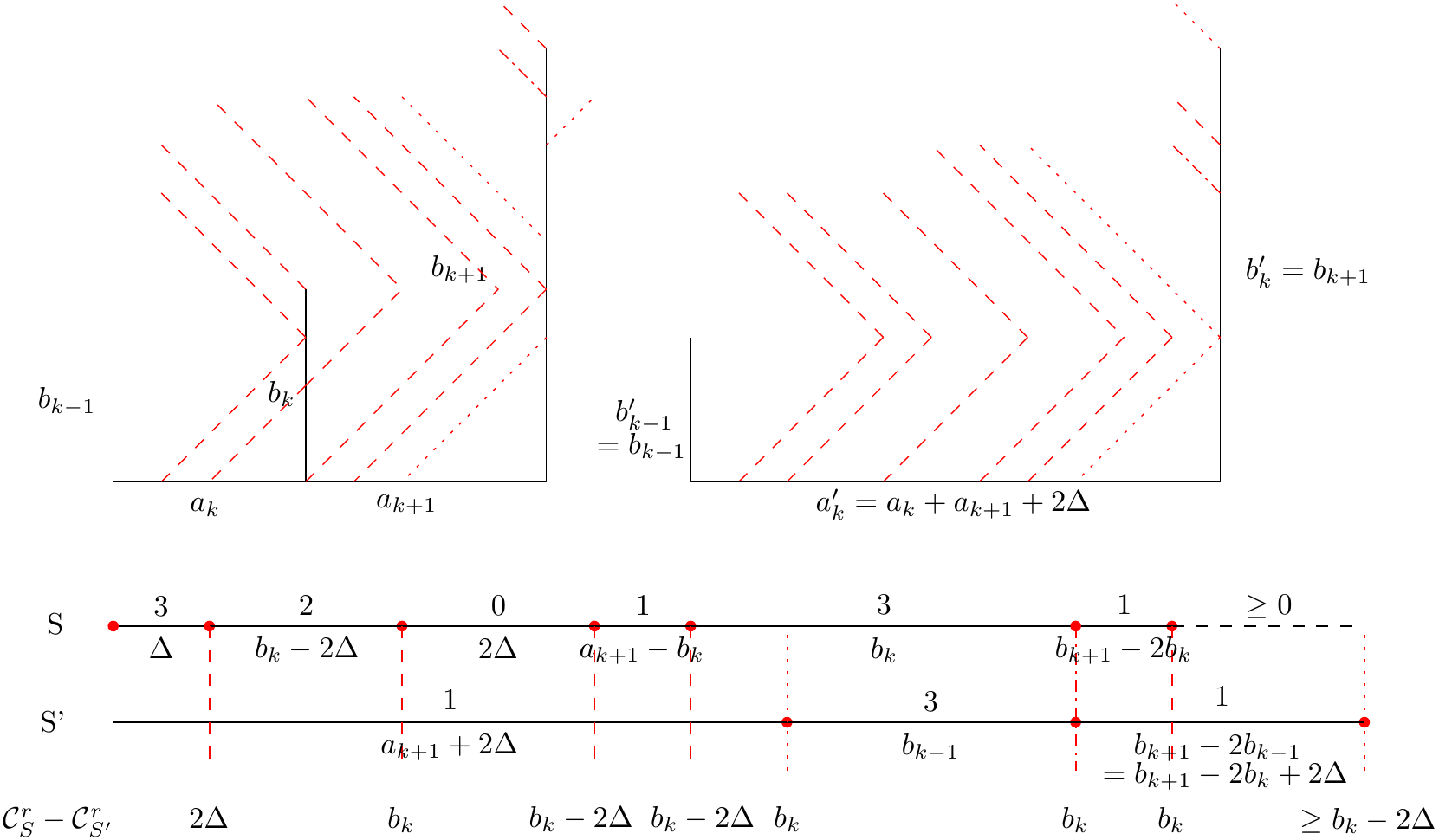}%
			\caption{ The situation in $\mathcal{S}$ and $\mathcal{S}'$ for $a_k \geq b_{k-1}$, $a_{k+1} \geq b_k$ and $b_{k+1} \geq 2b_k$ as well as the resulting intervals and their lengths. The red dashed and dotted lines indicate changes in \burn in either \crews.}
			\label{prereq-fig}
		\end{center}%
	\end{figure}
	To sketch the proof, let us assume that $a_k \geq b_{k-1}$, $a_{k+1} \geq b_k$ and $b_{k+1} \geq 2b_k$ hold.
	The \ratio in the right direction is identical for $\mathcal{S}$ and $\mathcal{S}'$ until time $A_k + b_{k-1}$ when the fire reaches $b_k$ in $\mathcal{S}$.
	\autoref{prereq-fig} shows the next sequences of consumption intervals in $\mathcal{S}$ and $\mathcal{S}’$ until the fire reaches	the top of $b'_k$ in $\mathcal{S}'$. Due to the linearity of \burn within each $k$-interval, only the
	points where intervals change can attain maximal values.
	Direct comparison shows that, $\mathcal{S}'$ has a smaller consumption at all such points in time.
	Once the fire has overcome the gap between $b'_{k-1}$ and $b'_k$ in $\mathcal{S}'$,
	each configuration $K'$ at time $t'$ in $\mathcal{S}'$ corresponds to a configuration $K$ at time $t = t-2 \Delta$ in $\mathcal{S}$. But, due to the presence of vertical barrier $b_k$ and the missing horizontal extension by $2 \Delta$, in $K$ the consumption differs by $b_k - 2 \Delta = 2 b_{k-1} - b_k$, which is positive by assumption.
	Thus, $K’$ has a lower consumption ratio than $K$.
	
	All other cases work similarly: The additional consumption contributed by the added $2\Delta$ of horizontal barrier between $\B'_{k-1}$ and $\B'_k$ is always covered by the removal of the vertical barrier of length $\B_k > 2 \Delta$.
	
	Note, that these arguments require that no part of $\A_k$ is covered by the head-start~\startWall. We can assume so by a similar argument. Let $\B_s$ ($s>1$) be the last vertical \wall in the right direction with $A_s \leq s$, which means that all horizontal barriers $\A_1, \A_2, \ldots, \A_s$ are covered by the head-start. Then combining all barriers $\B_1$ to $\B_s$ into one barrier $\B_s$ at the end of $\startWall$ does not increase $\BF{}^r(t)$ for any $t$. 
	
	This concludes the proof. \qed
\end{proof}

This means that when given an arbitrary \crews, we can assume $\B_i > 2 \B_{i-1}$ and $\D_i > 2 \D_{i-1}$ for all $i>1$. From this we can derive a helpful observation about the order of \burn of vertical and horizontal \walls in a \crews: when the fire reaches the top of a vertical \wall $b_i$ at some time $t$ (compare \autoref{RHsituation-fig}), every \wall $\A_k$ and $\B_k$ with $k \leq i$ has been completely consumed, as for every point on $\A_k$ or $\B_k$ the shortest non-barrier-crossing path has length smaller than $A_i + \B_i = t$. Hence, a $0$-interval in the right direction will begin at such times $t$ and $\BF{}^r(t) = A_i + B_i - s$, where \startWall denotes the length of the head-start not contributing to the consumption.
This observation holds equivalently for both directions.

\section{A lower bound of \ensuremath{v > 1.66}}\label{section:lowerBounds}

Assume there exists a \crews $\mathcal{S}$ consisting of horizontal \walls along the $x$-axis and vertical \walls attached to it. Further assume for $\mathcal{S}$ that $\BF{}(t) \leq v \cdot t$ at all times $t$ for some $v = (1+V)$ with $V \leq \frac{2}{3}$.
For this we will construct a contradiction by identifying a specific time $t_\mathcal{S}$, for which $\BF{}(t_\mathcal{S}) > (1+V) \cdot t_\mathcal{S}$.

By \autoref{lemma:prereq}, we can assume $\B_i > 2 \B_{i-1}$ and $\D_i > 2 \D_{i-1}$ for all $i>1$ in $\mathcal{S}$.
\begin{figure}[hbtp]%
	\begin{center}%
		\includegraphics[scale=0.5,keepaspectratio]{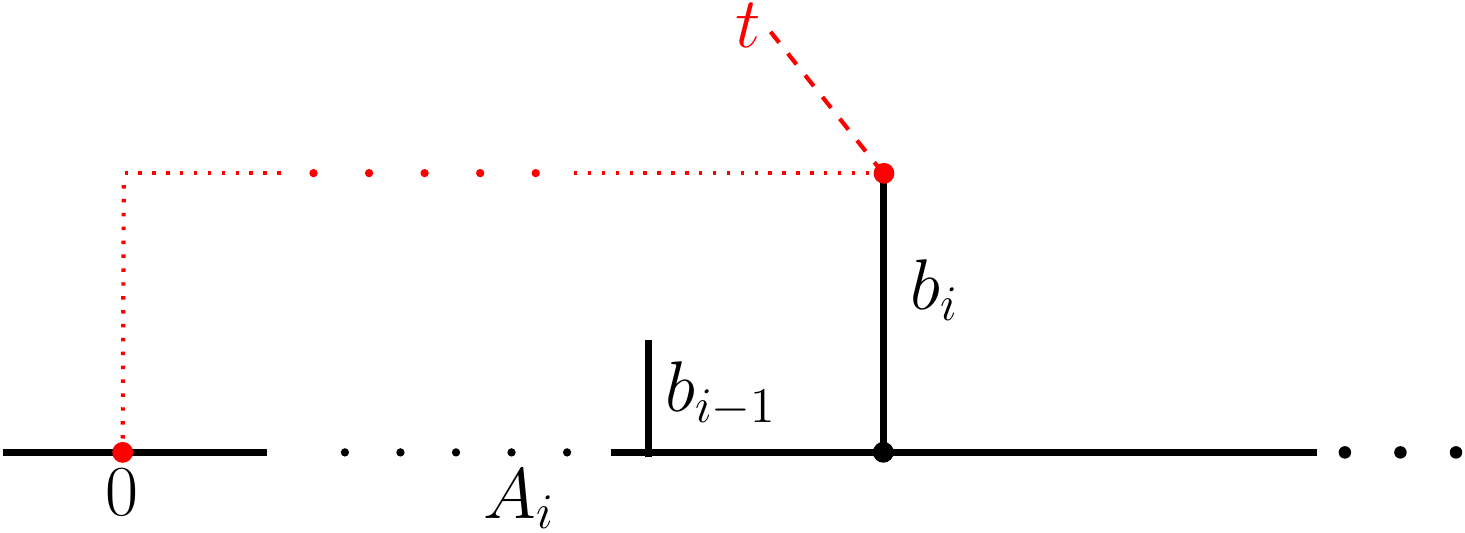}%
		\caption{At some time $t = A_i + \B_i$ the fire will reach the top of a vertical barrier $\B_i$.}
		\label{RHsituation-fig}
	\end{center}%
\end{figure}
As without vertical delaying barriers, the \ratio just goes towards $2$, $\mathcal{S}$ has an unbounded number of vertical barriers in at least one direction. \Wlog assume this is the right direction. 
Consider a moment when the fire reaches the end of some barrier $\B_i$ as illustrated in \autoref{RHsituation-fig}.
As explained in \autoref{sec:prereq}, this happens at time $t = \B_i + A_i$ and Lemma~\ref{lemma:prereq} implies we have $\BF{}^r(t) =  A_i + B_i - \startWall$.

\begin{eqnarray}
\BF{}^r(t) &= &  A_i + B_i - \startWall  = A_i + b_i + B_{i-1} - \startWall \quad\quad  \mid B_{i-1} > 2 \startWall \text{ for $i$ large enough}\nonumber\\
&>&  A_i + \B_i + \startWall > t + \startWall > t \label{leftRequirement}
\end{eqnarray}
Hence for $t$ large enough, $\QF{}^r(t) > 1$ at times $t$, when the fire reaches the top of a vertical barrier.
Therefore, $\mathcal{S}$ has repeated $0$-intervals in the left direction as well, or else $\QF{}^l(t)$ would go towards $1$ and $\QF{}(t) > 2$ at such times $t$.
\begin{figure}[hbtp]%
	\begin{center}%
		\includegraphics[width=\textwidth,keepaspectratio]{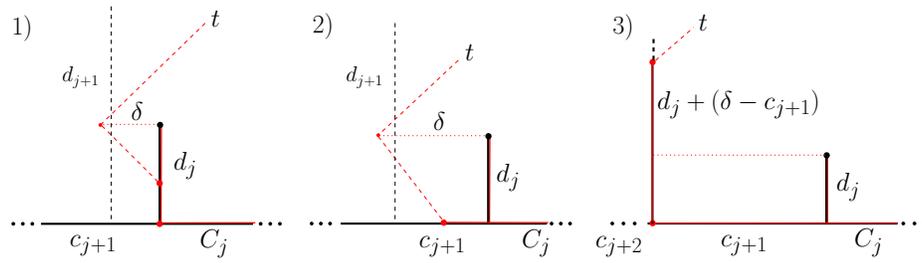}%
		\caption{All three possible situations for the left side to be in at time $t$. Note that in case 1) and 2) the fire might have reached $\D_{j+1}$, which does not affect our considerations.}
		\label{AllCases-fig}
	\end{center}%
\end{figure}
We now consider the situation in the left direction at time $t = \B_i + A_i$. Let $\D_j$ denote the last vertical barrier, whose upper end was reached by the fire, so $t = \D_j + C_j + \delta$ with $0 \leq \delta < \C_{j+1} + \D_{j+1} - \D_{j}$. 
\Wlog we assume that $\B_{i+1} + A_{i+1} \geq \D_{j+1} + C_{j+1}$. Otherwise, there must be multiple vertical barriers in the right direction whose upper ends are reached by the fire after it reaches the upper end of $\D_j$ and before it reaches the upper end of $\D_{j+1}$. In that case, we can assume that $\B_i$ is the last among those, such that $\B_{i+1} + A_{i+1} \geq \D_{j+1} + C_{j+1}$ holds.

We split our consideration in three cases, which are all illustrated in \autoref{AllCases-fig}:
\begin{enumerate}
	\item $0 \leq \delta < \D_j$ \label{LH1} 
	\item $\D_j \leq \delta < \D_j + \C_{j+1}$ \label{LH2} 
	\item $\D_j + \C_{j+1} \leq \delta < \C_{j+1} + \D_{j+1} - \D_{j}$ \label{LH3}
\end{enumerate}

In the first case, the fire has not reached the horizontal barrier $\C_{j+1}$ yet after passing over $\D_j$; in the second case, it has reached $\C_{j+1}$, but not its end; in the third case the fire has completely consumed $\C_{j+1}$.

In Case~\ref{LH3}, $\delta = \D_j + \C_{j+1} + \epsilon$ and then $\BF{}^l(t) \geq C_{j+1} + D_{j} + 2\D_j + \epsilon - \startWall > (\D_j + C_j) + (\D_j + \C_{j+1}) + \epsilon = t$, which together with Inequality~(\ref{leftRequirement}) already gives $\BF{} (t) > 2 t > (1+V) \cdot t$ which is a contradiction.

For both remaining cases, we will derive a lower bound for $\D_j$. We will then consider the moment $t_1 =  2 \D_j + C_{j+1}$, when the fire reaches the end of the horizontal barrier $\C_{j+1}$.
Using the lower bound on $\D_j$, we will prove $\BF{}(t_1) > (1+V) \cdot t_1$.

\subsection{Case 1: \ensuremath{0 \leq \delta < \D_j}} 
In Case \ref{LH1},  
 $\BF{}^l(t) > C_{j} + D_{j} - \startWall = C_{j} + \D_{j} + D_{j-1} - \startWall > C_{j} + \D_{j}$, since $D_{j-1} > s$ for $j$ large enough.
Now at time $t$, it must hold:
\begin{eqnarray}
\BF{}(t) = \BF{}^r(t) + \BF{}^l(t) & < &  (1+V) \cdot t \quad\quad\quad\quad\quad\quad\,\,\,\, \mid \text{Inequality~(\ref{leftRequirement})}\nonumber \\
\Rightarrow  \,\quad\quad\quad\quad\quad\quad C_{j} + \D_{j} & < & V (\D_j + C_{j} + \delta) \nonumber \\
\Rightarrow  \quad\quad\quad\quad\quad (V-1) C_{j} & >& (1-V) \D_j - V \delta \quad\quad\quad\quad \mid \left(V < 1 \right)   \nonumber \\
\Leftrightarrow \,\,\quad\quad\quad\quad \quad\quad\quad\quad  C_j & < & \nicefrac{V}{(1-V)} \cdot \delta - \D_j \label{lastIn1}
\end{eqnarray}
$V\leq\frac{2}{3}$ implies $\nicefrac{V}{(1-V)} \leq 2$ by direct calculation, which gives bounds for $C_j, \D_j$:
\begin{eqnarray}\label{central}
C_j & < & 2 \delta - \D_j < \D_j \quad \mid \delta < \D_j \text{ in Case~\ref{LH1}}\nonumber\\
\Rightarrow \,\,\,\quad 2 \D_j &>& C_j + \delta \nonumber \\
\Leftrightarrow \quad\quad \D_j &>& \nicefrac{1}{2} (C_j + \delta) \label{LH1-D-Bound}
\end{eqnarray}

\subsection{Case 2: \ensuremath{\D_j \leq \delta < \D_j + \C_{j+1}}} 
In Case~\ref{LH2} a part of $\C_{j+1}$ of length $(\delta-\D_j)$ has already been consumed, so $\BF{}^l(t) \geq D_j + C_j + (\delta-\D_j) - \startWall > \D_j + C_j + (\delta-\D_j) = C_j + \delta$, as $D_{j-1} > s$ for $j$ large enough.
Now at time $t$ it must hold
\begin{eqnarray}
\BF{}(t) = \BF{}^r(t) + \BF{}^l(t) & < &  (1+V) \cdot t \quad\quad\quad\quad\quad\quad\,\,\,\, \mid \text{Inequality~(\ref{leftRequirement})}\nonumber \\
\Rightarrow \quad\quad\quad\quad\quad\:\:\:\:\: C_j + \delta &  <  & V( \D_j + C_j + \delta) \nonumber\\
\Rightarrow \quad\quad (1-V) (C_j + \delta) & < & V \D_j \nonumber\\
\Rightarrow \quad\quad\quad\quad\quad\quad\quad\quad \D_j & > & \nicefrac{(1-V)}{V} ( C_j + \delta )
\end{eqnarray}
$V\leq\frac{2}{3}$ implies $\nicefrac{(1-V)}{V} \geq \frac12$ by direct calculation, which gives the bound:
\begin{equation}
\D_j  >   \nicefrac{1}{2}(C_j + \delta) \label{LH2-D-Bound}
\end{equation}
This is the same bound as found for Case~\ref{LH1} in Inequality~(\ref{LH1-D-Bound}).

\subsection{Deriving the contradiction \ensuremath{\BF{}(t_1) > (1+V) \cdot t_1}}
\begin{figure}[hbtp]%
	\begin{center}%
		\includegraphics[scale=0.5,keepaspectratio]{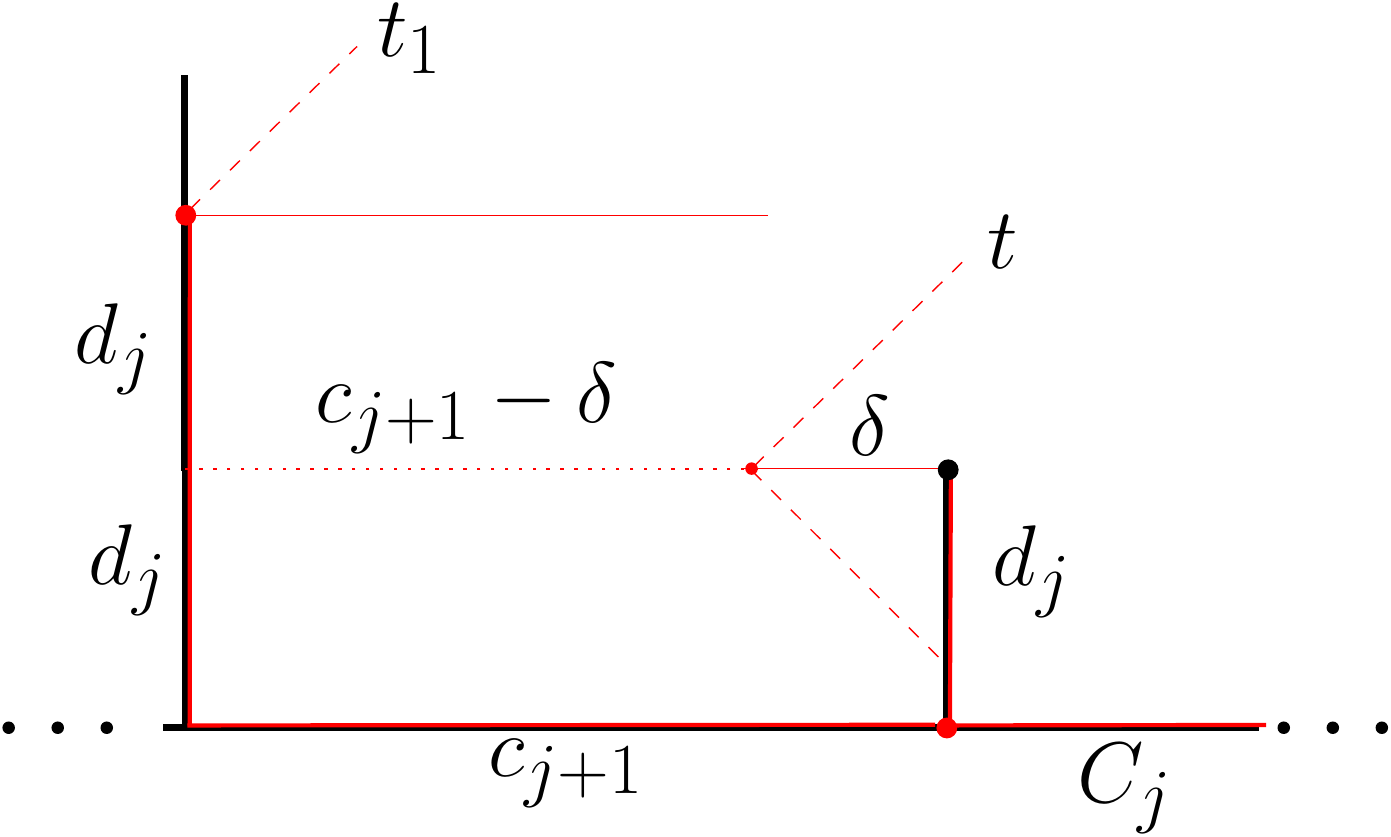}%
		\caption{After $\D_{j} + \C_{j+1}  - \delta$ additional time after $t$, the fire has reached the end of $\C_{j+1}$ and has also consumed a piece of length $2\D_j$ of the next vertical barrier.}
		\label{LH1situation-fig}
	\end{center}%
\end{figure}
Now we consider time $t_1 = C_{j+1} + 2 \D_j > t$, when the fire reaches the end of the horizontal \wall $\C_{j+1}$.
As for any time, at time $t_1$, it must hold 
\begin{eqnarray}
\BF{}(t_1) &=& \BF{}^r(t_1) + \BF{}^l(t_1) \leq (1+V) \cdot t_1 \nonumber \\
\Leftrightarrow \quad\quad\:\:\: \BF{}^l(t_1) &\leq & (1+V) \cdot t_1 - \BF{}^r(t_1) \nonumber\\
&=& (1+V) \cdot t + (1+V) (t_1-t) - (\BF{}^r(t) + \BF{}^r(t, t_1)) \nonumber\\ 
&\leq& V t + (1+V) (t_1-t) + t - \BF{}^r(t) \quad \mid \text{Ineq.~(\ref{leftRequirement})}\label{TightRequirement} \\
\Rightarrow \quad \BF{}^l(t_1) + s  &<& V t + (1+V) (t_1-t)\label{Case1leftRequirement}
\end{eqnarray}
By construction, $t_1 = C_{j+1} + 2 \D_j$.
As $t = \D_j + C_{j} + \delta$, this means $t_1 = t + (\D_j + \C_{j+1} - \delta)$.
Due to Lemma~\ref{lemma:prereq}, we know that the fire has not reached the end of $\D_{j+1}$ yet, hence $\BF{}^l(t_1) \geq 3 \D_j + C_{j+1} -\startWall$. 
Hence, we arrive at the following inequalities:
\begin{eqnarray}
3 \D_j + C_{j+1} & < &   V(\D_j + C_{j} + \delta) +  (1+V) (\D_j + \C_{j+1}-\delta) \nonumber \\
\Leftrightarrow \quad -V (\C_{j+1}-\delta) & < & (V-1)\delta + (V-1) C_{j} + (2V-2)\D_j \quad\quad\quad \mid \left(1>V\right) \nonumber\\
\Leftrightarrow \quad\quad\quad\,\, \C_{j+1}-\delta & > & \frac{1-V}{V} \delta + \frac{1-V}{V} C_{j} + 2\frac{1-V}{V} \D_j \,. \label{last2}
\end{eqnarray}
$V\leq\frac{2}{3}$ implies $\nicefrac{(1-V)}{V} \geq \frac12$ by direct calculation, which gives the bound:
\begin{eqnarray}
\C_{j+1}-\delta & > &  \frac{1}{2} \delta + \frac{1}{2} C_j + \D_j\, \nonumber\\
\Leftrightarrow \qquad \D_j + \C_{j+1} - \delta &>& \frac{1}{2} \delta + \frac{1}{2} C_j + 2 \D_j\label{finalconcl} 
\end{eqnarray}
Now in both cases we got $\D_j > \nicefrac{1}{2} (C_j + \delta)$ (Inequalities~(\ref{LH1-D-Bound}) and (\ref{LH2-D-Bound})), so we can apply that and conclude: 
\begin{eqnarray}
t_1-t = \D_j + \C_{j+1} - \delta & > & C_j + \D_j + \delta = t = A_i + \B_i \label{ineq:final}
\end{eqnarray}  

So we know, that in both cases $t_1 - t > \B_i + A_i$.
Now consider the situation in the right direction again (compare \autoref{RHsituation-fig}). At $t + \B_i$ the fire reaches the horizontal \wall $\A_{i
+1}$ behind $\B_i$. Additionally, by assumption $\B_{i+1} + A_{i+1} \geq \D_{j+1} + C_{j+1}$, the fire has not reached the top of the next barrier $\B_{i+1}$ at $t_1$. This means, that between $t + \B_i$ and $t_1$, there is always at least consumption $1$ in the right direction, which means the fire has consumed \walls of length at least $A_i$, hence $\BF{}^r(t,t_1) \geq A_i$. 

As our whole consideration is based on inequalities, we will consider an edge case with a contradiction that can be extended to our given \crews $\mathcal{S}$.
More precisely, assume, that Inequality~(\ref{TightRequirement}) is tight for some $t_1^*$, so:
\begin{eqnarray*}
\BF{}^l(t_1^*) &=& V t + (1+V) (t_1^*-t) + t - \BF{}^r(t)  \\
\Leftrightarrow \quad \BF{}^r(t) + \BF{}^l(t_1^*) &=& (1+V)t_1^*
\end{eqnarray*}
By our arguments above, $\BF{}^r(t,t_1^*) \geq A_i$ and hence $\BF{}(t_1^*) = \BF{}^r(t,t_1^*) + \BF{}^r(t) + \BF{}^l(t_1^*) \geq (1+V)t_1^* + A_i > (1+V)t_1^*$, which is a contradiction for this edge case.
 
Now in our given \crews $\mathcal{S}$ it holds $t_1 = t_1^* + x$ for some $x>0$. As everything except $\C_{j+1}$ is fixed at $t$, this additional time results in additional consumption of at least horizontal \walls of length $x$ in both directions in comparison to the edge case. Hence we can extend the contradiction:
\begin{eqnarray*}
\BF{}(t_1) &=& \BF{}^l(t_1) + \BF{}^r(t) + \BF{}^r(t,t_1) \\
&=& \BF{}^l(t_1^*) + \BF{}^r(t) + \BF{}^r(t,t_1^*) + 2x \\
&=& (1+V)t_1^* + 2x + A_i > (1+V) (t_1^* + x) = (1+V) t_1.
\end{eqnarray*}

\begin{theorem}
	The fire can not be contained in the upper half-plane with speed $v \leq 1.66$ by a \crews consisting of a horizontal \wall along the $x$-axis and vertical \walls attached to it.
\end{theorem}

\section{Upper bounds}
\label{section:upperBounds}
We prove the upper bound by defining a \crews with bounded \ratio.
Before we present the construction, we give some intuition.
We choose the following conditions:
\begin{equation}
\label{equation:conditionsOfRecursion}
\begin{array}{r c c c l}
		& \A_{i+1} \geq \B_{i} 	& \text{and} 	& \B_{i+1} \geq 2 \B_{i}  	&	\forall i \geq 1, \\
\text{similarly } 
		& \C_{i+1} \geq \D_{i} 	& \text{and}	& \D_{i+1} \geq 2 \D_{i}  &	\forall i \geq 1.
\end{array}
\end{equation}
This forces the $0$-intervals generated by $b_i$ to be of length of $b_i$.
For a single direction this results in a repeating sequence of $k$-intervals of specific lengths and $k$ as shown in \autoref{fig:interval-cycle}.
\begin{figure}[ht]
	\centering
	\includegraphics[width=\textwidth]{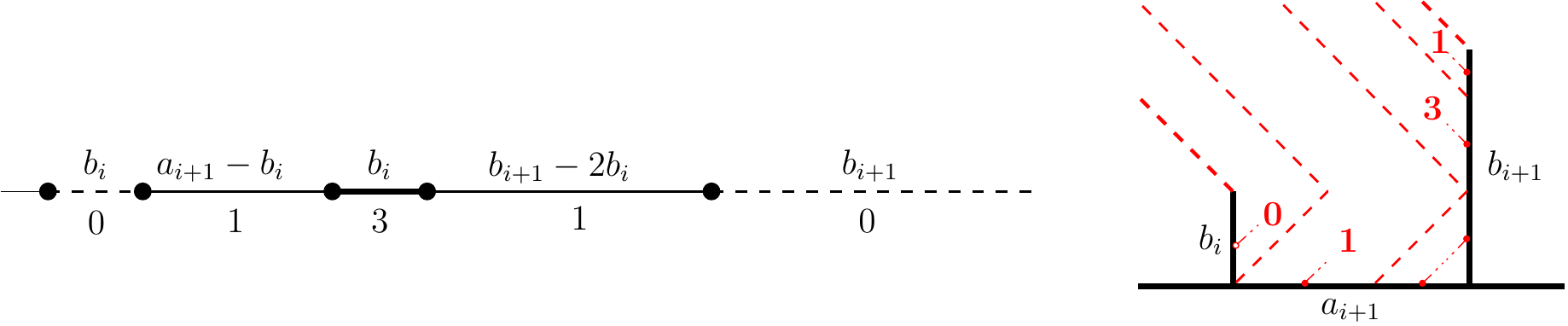}
	\caption{
	A sequence of $k$-intervals to the right of $(0,0)$.
	The length is given above each interval and the current \burn below.
	}
	\label{fig:interval-cycle}
\end{figure}
The idea is to construct the \crews in such a way that the $0$-intervals always appear in an alternating fashion, so the local maxima in the \ratio of one direction can be countered by the $0$-intervals of the other direction.

To show that this idea can be realized, we consider the periodic interlacing of time intervals as illustrated in \autoref{fig:recursiveInterlacing}.
There, the ends of the $0$-intervals in one direction coincide with the ends of the \mbox{$3$-intervals} in the other direction, that is, at $t_3$ and $t_6$.
\begin{figure}[ht]
	\centering
	\includegraphics[width=\textwidth]{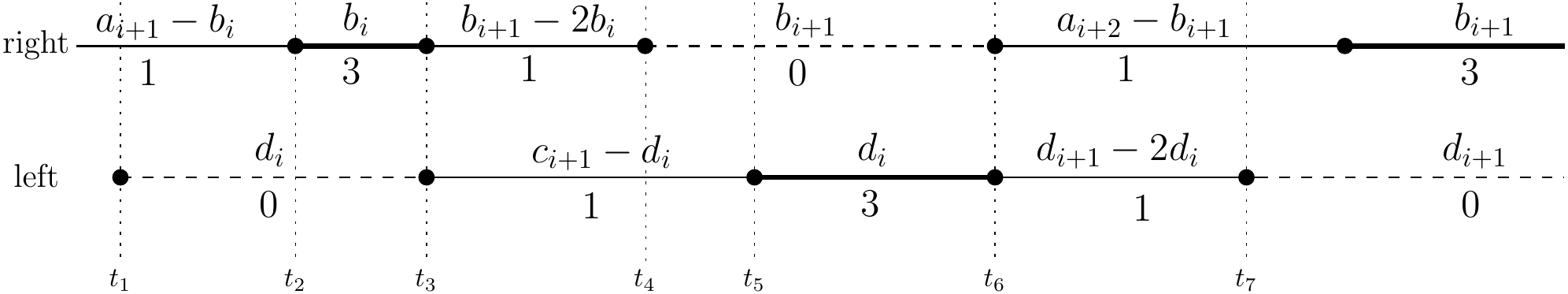}
	\caption{The periodic interlacing of time intervals.}
	\label{fig:recursiveInterlacing}
\end{figure}

The current \burn is always greater than $1$, since the $0$-intervals do not overlap.
Also, the combined \ratio $\QF{}(t)$ must be smaller than $2$ at all times.
This also implies that $t_3$ is no local maximum and the \ratio grows towards 2 between $t_3$ and $t_4$.
Hence, by setting $d_i > 2b_i$ we make $t_1$, $t_4$, $t_7$ the local maxima and $t_2$, $t_5$ the local minima of $\QF{}(t)$.

\newcommand{\constShift}{\ensuremath{\alpha}\xspace}
\newcommand{\constGrowth}{\ensuremath{\beta}\xspace}
Let us now consider the \ratio $\QF{}(t_1, t_4)$ of the cycle from $t_1$ to $t_4$.
There are two $1$-intervals involved in this cycle in the right direction.
The first one, where the fire burns along $a_{i+1}$, is of length $a_{i+1}-b_i$ and lies partially in this cycle.
The second one, where the fire crawls up along $b_{i+1}$, is of length $b_{i+1}-2b_i$ and lies completely in this cycle.
As the beginning of this cycle is given by the start of the $0$-interval on one side and the end is given by the end of the second $1$-interval on the other side, we know that the length of this cycle is $d_i + (b_{i+1} - 2 b_i)$.
The total \burn in this cycle is $1 \cdot d_i + 2 \cdot b_i + 2(b_{i+1} - 2 b_i)$.
Now we define $d_i = \constGrowth\cdot b_{i}$, $b_{i+1} = \constGrowth \cdot d_i$, and $d_i = \constShift + 2b_i$ for some $\constShift, \constGrowth \in \mathbb{R}_{>0}$.
Note that this choice satisfies all our conditions, including $d_i > 2 b_i$, and that $\constShift = (\constGrowth -2)b_i$ and $b_{i+1} = \constGrowth ^2b_i$.
Then the \ratio $\QF{}(t_1, t_4)$ of the cycle is given by
\begin{equation*}
	\frac{\BF{}(t_1, t_4)}{t_4 - t_1}
	= \frac{(\constShift + 2b_i) + 2b_i + 2(b_{i+1}-2b_i)}{(\constShift + 2 b_i) +b_{i+1}-2b_i} 
	= \frac{\constShift + 2b_{i+1}}{\constShift + b_{i+1}}
	= \frac{(\constGrowth-2) + 2\constGrowth^2}{(\constGrowth -2)+ \constGrowth^2}
\end{equation*}
and attains a minimal value of \nf{17}{9} for $\constGrowth = 4$.
Note that by design, $\QF{}(t_1, t_2)$ and $\QF{}(t_1, t_3)$ stay below \nf{17}{9}, as well.
Moreover, if the \ratio has a maximum of \nf{17}{9} at the beginning of the cycle at $t_1$, this will also be the case at the end at $t_4$ as
\begin{equation*}
\QF{}(t_4) 
	= \frac{\BF{}(t_1)+\BF{}(t_1, t_4)}{t_4}
	= \frac{t_1}{t_4} \cdot \frac{\BF{}(t_1)}{t_1} + \frac{t_4 - t_1}{t_4} \cdot \frac{\BF{}(t_1, t_4)}{t_4 - t_1}
	\leq \frac{17}{9}.
\end{equation*}
Since the cycles change their roles at $t_4$ such that the $0$-interval occurs on the right side of $(0,0)$, the same argument can be used to bound the local \ratio in the following interval and for all subsequent cycles, recursively.
Note that by looking at the time interval from $t_3$ to $t_6$, we can derive a closed form for $c_{i+1}$.
Similarly we proceed for $a_{i+1}$.

To prove the final theorem, it remains to find initial values to get the interlacing started, while maintaining $\QF{}(t) \leq \nf{17}{9}$.
Suitable values are
\begin{equation*}
\begin{array}{r l  r l  r l r l r l}
\A_1 &:= \startWall 
	& \hspace{.5cm}\B_1  	&:= 17 \startWall
	& \hspace{.5cm} \A_2		&:= 34 \startWall
	& \hspace{.5cm}\A_{i+1} &:= 7.5 \B_{i}
	& \hspace{.5cm}\B_{i+1} &:= 4 \D_{i}\\
\C_1 & := \startWall
	& \D_1		&:= 34 \startWall 
	& \C_2 		&:= 238 \startWall
	& \C_{i+1} 	&:= 7.5 \D_{i}
	& \D_{i+1} 	&:= 4 \B_{i+1},\\
\end{array}
\end{equation*}
which results in the starting intervals given in \autoref{fig:intervalStart}.
The local maxima at $t_1$ and $t_4$ then have \ratio exactly $\nf{17}{9}$. The interval between $t_2$ and $t_3$ is set up equivalent to the one between $t_3$ and $t_6$ in \autoref{fig:recursiveInterlacing}, which means the interlacing construction can be applied to all intervals beyond.
Note that all barriers scale with \mbox{$s$.}
An example of this construction for $\startWall = 1$ is given in \autoref{fig:exampleBarrierSystem}.
\begin{figure}[ht]
	\centering
	\includegraphics[width=.8\textwidth]{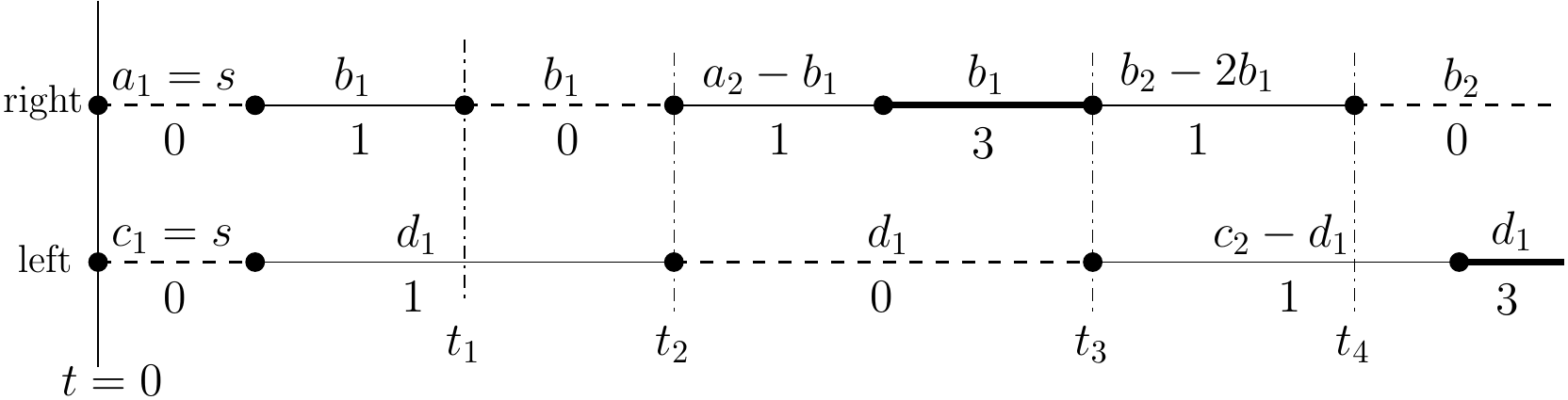}
	\caption{Illustration of time intervals at the start. Due to their growth, the sizes of the intervals are not true to scale.}
	\label{fig:intervalStart}
\end{figure}
\begin{figure}[ht]
	\centering
	\includegraphics[width=.8\textwidth]{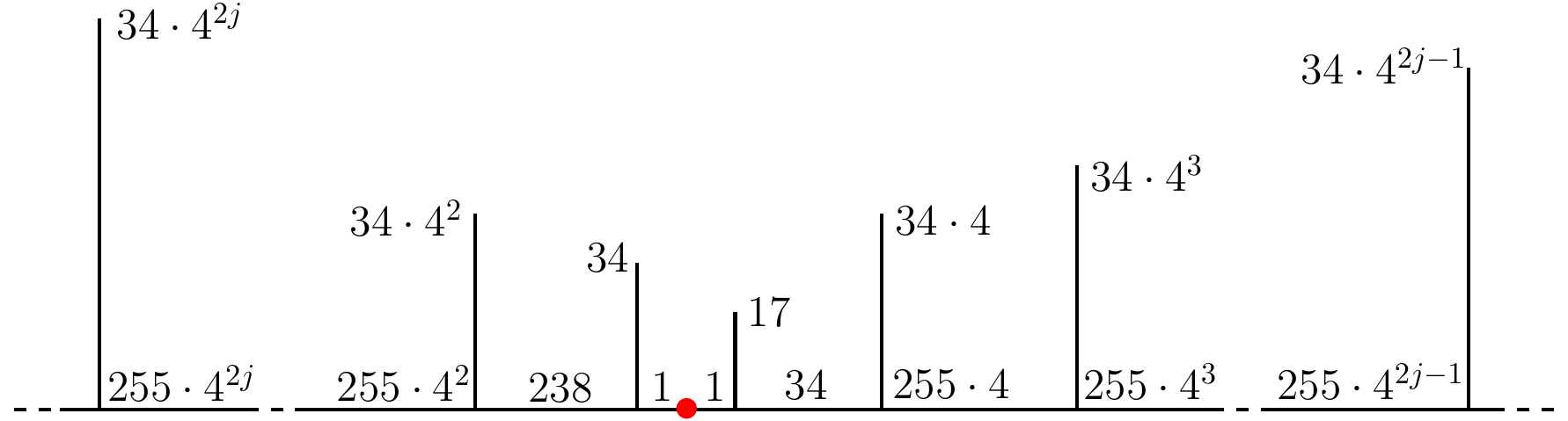}
	\caption{Example for the final \crews for $\startWall = 1$, also not true to scale.}
	\label{fig:exampleBarrierSystem}
\end{figure}
\begin{theorem}
The fire can be contained in the upper half-plane with speed $v = \frac{17}{9} =  1.\overline{8}$
\end{theorem}

\newcommand{\constOffset}{\ensuremath{\delta}\xspace}
\subsection{Improving the upper bound}
It is possible to reduce the upper bound of $v=1.\overline{8}$ slightly.
As shown in \autoref{fig:recursiveInterlacing}, the end of the $3$-interval in one direction coincides with the end of the $0$-interval in the other direction, which makes $t_4$ the only local maximum of the interval $[t_1, t_4]$.
We introduce a regular shift by a factor of \constOffset , see \autoref{fig:recursiveInterlacingImproved}.
This allows the $3$-interval in one direction to lie completely inside the $0$-interval of the other direction, as shown in \autoref{fig:recursiveInterlacingImproved}.
Then, there are two local maxima in the equivalent interval $[t_1, t_5]$, namely at $t_3$ and $t_5$.
We force both maxima to attain the same value to minimize both at the same time.

\begin{figure}[ht]
	\centering
	\includegraphics[width=\textwidth]{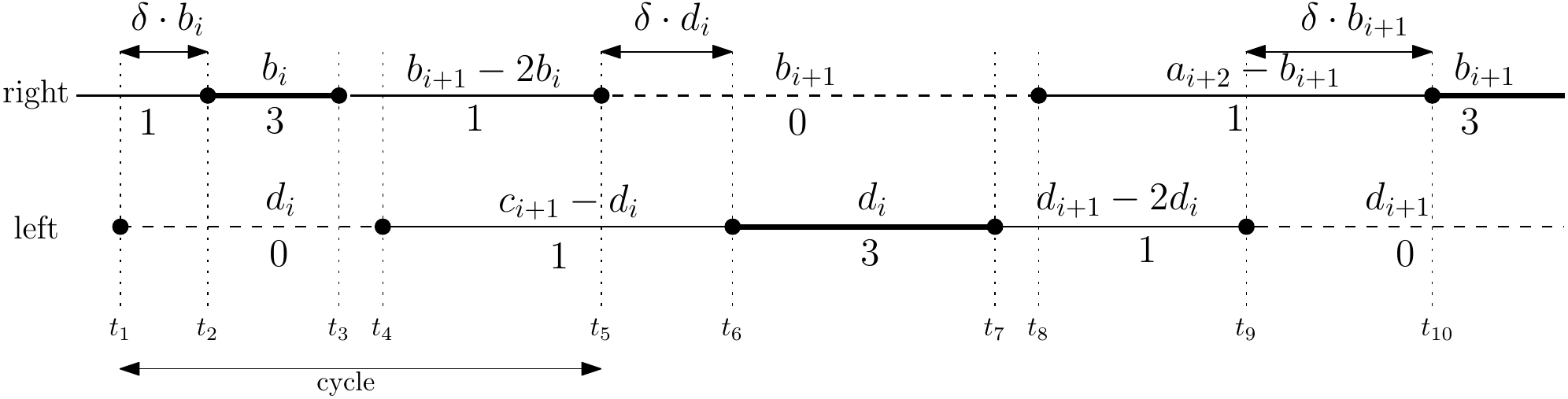}
	\caption{A general periodic interlacing of time intervals.}
	\label{fig:recursiveInterlacingImproved}
\end{figure}
Again, we set $d_i = \constGrowth \cdot b_i$ and $b_{i+1} = \constGrowth \cdot d_i$, for some $\constGrowth \geq 1$ determined below.
Then the value of the first local maximum can be expressed as
\begin{equation*}
\QF{}(t_1, t_3) 
	= \frac{\BF{}(t_1, t_3)}{t_3 - t_1}
	= \frac{1 \cdot (\constOffset \cdot b_i) + 3 \cdot b_i}{\constOffset \cdot b_i + b_i} 
	= \frac{\constOffset + 3}{\constOffset + 1} 
	= 1 + \frac{2}{\constOffset + 1}.
\end{equation*}
Considering the cycle from $t_1$ to $t_5$ in \autoref{fig:recursiveInterlacingImproved}, we can conclude that $c_{i+1} = b_{i+1} - b_i + \constOffset b_i + \constOffset d_i$.
Similarly, we can proceed on the interval from $t_5$ to $t_9$ to express $a_{i+1}$ in terms of \constGrowth, \constOffset and $b_i$.

Using these identities, we obtain for the second local maximum
\begin{equation*}
\def\arraystretch{1.4}
\begin{array}{rl}
\QF{}(t_1, t_5) 
	& = \frac{\BF{}(t_1, t_5)}{t_5 - t_1}
		= \frac{1 \cdot (\constOffset \cdot b_i) + 3b_i + 1\cdot (b_{i+1} - 2 b_i ) + 1 \cdot (c_{i+1} - d_i - \constOffset \cdot d_i)}{d_i + (c_{i+1} - d_i) - \constOffset \cdot d_i}\\
	& = \frac{c_{i+1} - \constOffset \cdot d_i}{c_{i+1} - \constOffset \cdot d_i} + \frac{b_{i+1} + \constOffset \cdot b_i + b_i - d_i}{c_{i+1} - \constOffset \cdot d_i}
	  = 1 + \frac{b_{i+1} - b_i + \constOffset \cdot b_i + 2b_i - d_i}{b_{i+1} - b_i + \constOffset b_i}\\
	& = 2 + \frac{2 b_i - d_i}{b_{i+1} - b_i + \constOffset b_i} = 2 + \frac{2 - \constGrowth}{\constGrowth^2 - 1 + \constOffset}.
\end{array}
\end{equation*}
As mentioned above, we set both local maxima to be equal, solve for \constOffset and obtain
\begin{equation*}
\delta = \frac{1}{2} \left( \constGrowth - \constGrowth^2 + \sqrt{-12 + 4 \constGrowth + 5 \constGrowth^2 - 2 \constGrowth^3 + \constGrowth^4}\right).
\end{equation*}
Plugging this into either one of the two local maxima and minimizing the resulting function for $\constGrowth\geq 1$, we obtain
\begin{equation*}
\constGrowth = \frac{3}{2} + \frac{1}{6} \left( 513 - 114 \sqrt{6} \right)^{\nf{1}{3}} + \frac{\left( 19 (9 + 2 \sqrt{6}) \right)^{\nf{1}{3}}}{2 \cdot 3^{\nf{2}{3}}} \approx 4.06887 
\end{equation*}
for the optimal value of \constGrowth, $\constOffset \approx 1.2802$ and
\begin{equation*}
v = \frac{1}{6} \left( 10 - \frac{19^{\nf{2}{3}}}{\sqrt[3]{2 (4 + 3 \sqrt{6})}} + \frac{\sqrt[3]{19 (4 + 3 \sqrt{6}))}}{2^{\nf{2}{3}}} \right) \approx 1.8771 
\end{equation*}
as the minimum speed.

Note that the optimal value for \constGrowth satisfies our conditions given in \autoref{equation:conditionsOfRecursion}, so that the \crews can in fact be realized.
Finally, we give suitable values to get the interlacing started:
\begin{equation*}
\begin{array}{r l  r l}
\B_1  		& := 1
	& \D_1  		&:= 2 \B_1\\
\startWall 	& := \frac{(4\constGrowth + 2 \delta + 1) - v(2\constGrowth + \delta + 1)}{v} \cdot \B_1
	& \A_1 & := \C_1 := \startWall\\[1em]
\A_2  		& := (\delta + 1) \cdot \B_1
	& \C_2  		&:= (2 \constGrowth + 3 \delta - 1) \cdot \B_1\\
\A_{i+1} &:= (\delta - 1) \D_i + (\constGrowth + \delta) \B_{i+1}
	& \hspace{1cm}\B_{i+1} &:= \constGrowth \cdot \D_i\\ 
\C_{i+1} &:= (\delta - 1) \B_i + (\constGrowth + \delta) \D_i
	& \D_{i+1} &:= \constGrowth \cdot \B_{i+1}.\\ 
\end{array}
\end{equation*}
To keep the expression simple, we fixed the value of $\B_1$ and scaled the value of $\startWall$ as listed above.
These values can be rescaled to work for any given \startWall.
\begin{theorem}
The fire can be contained in the upper half-plane with speed $v = 1.8772$.
\end{theorem}

\section{Conclusion}\label{section:conclusion}
We have shown non-trivial bounds for the problem of protecting the lower half-plane from fire with an infinite horizontal \wall.
Our results show that delaying barriers -- in this case vertical segments attached to the horizontal \wall -- can help to break the obvious upper bound of~2 for the building speed.
More complex delaying barriers, \eg free-floating ones, were not analysed specifically, however it is hard to imagine a way for those to have improving effects.
It will be interesting to see if such an effect can also be achieved for the problem of containing the fire by a closed barrier curve, \ie, for Bressan's original problem.
As a intermediate result in that direction, one ought to extend these results to the Euclidean metric first, where the effect of delaying barriers is less pronounced and harder to analyse. 

\paragraph*{Acknowledgements}
We thank the anonymous referees for their valuable input.

\footnotesize
\bibliographystyle{splncs04}
\bibliography{includes/bibliography}

\end{document}